\def\year{2020}\relax
\documentclass[letterpaper]{article} % DO NOT CHANGE THIS
\usepackage{aaai20}  % DO NOT CHANGE THIS
\usepackage{times}  % DO NOT CHANGE THIS
\usepackage{helvet} % DO NOT CHANGE THIS
\usepackage{courier}  % DO NOT CHANGE THIS
\usepackage[hyphens]{url}  % DO NOT CHANGE THIS
\usepackage{graphicx} % DO NOT CHANGE THIS
\usepackage{comment}
\usepackage{graphicx}  % DO NOT CHANGE THIS
\usepackage{amsmath,amsthm}
\usepackage{amssymb}
\newcommand{\citet}[1]{\citeauthor{#1} \shortcite{#1}}
\newcommand{\citep}{\cite}

\newtheorem{theorem}{Theorem}

\nocopyright
\title{Solving Online Threat Screening Games using Constrained\\ Action Space Reinforcement Learning}
\author{Sanket Shah,\textsuperscript{\rm 1} Arunesh Sinha,\textsuperscript{\rm 1} Pradeep Varakantham,\textsuperscript{\rm 1} Andrew Perrault,\textsuperscript{\rm 2} Milind Tambe\textsuperscript{\rm 2}\\ % All authors must be in the same font size and format. Use \Large and \textbf to achieve this result when breaking a line
\textsuperscript{\rm 1}School of Information Systems,  Singapore Management University,  %If you have multiple authors and multiple affiliations
\{sankets, aruneshs, pradeepv\}@smu.edu.sg \\ % email address must be in roman text type, not monospace or sans serif
\textsuperscript{\rm 2}Harvard University, \{aperrault@g., milind\_tambe@\}harvard.edu
}

\usepackage{tikz}
\begin{document}

\maketitle

\begin{abstract}
Large-scale screening for potential threats with limited resources and capacity for screening is a problem of interest at airports, seaports, and other ports of entry. Adversaries can observe screening procedures and arrive at a time when there will be gaps in screening due to limited resource capacities. To capture this game between ports and adversaries, this problem has been previously represented as a Stackelberg game, referred to as a Threat Screening Game (TSG).  Given the significant complexity associated with solving TSGs and uncertainty in arrivals of customers, existing work has assumed that screenees arrive and are allocated security resources at the beginning of the time window. In practice, screenees such as airport passengers arrive in bursts correlated with flight time and are not bound by fixed time windows. To address this, we propose an online threat screening model in which screening strategy is determined adaptively as a passenger arrives while satisfying a hard bound on acceptable risk of not screening a threat. To solve the online problem with a hard bound on risk, we formulate it as a Reinforcement Learning (RL) problem with constraints on the action space (hard bound on risk). We provide a novel way to efficiently enforce linear inequality constraints on the action output in Deep Reinforcement Learning. We show that our solution allows us to significantly reduce screenee wait time while guaranteeing a bound on risk.
\end{abstract}

\section{Introduction}
Screening for potential threats entering large safety-sensitive establishments (e.g., airports, seaports, museums) using the right subset of available screening methods (e.g., metal detectors, advanced imaging technology, pat-down) is an important defensive activity undertaken by various agencies around the world. However, the sheer scale of the problem at these large establishments with a large number of screenees and various screening methods makes screening in a timely fashion with limited resources quite challenging. For example, average delays for air passengers in Chicago, USA jumped to 2 hours due to higher passenger volume in the summer of 2016~\cite{reuters}. Additionally, intelligent adversaries can exploit any gaps in screening and cause catastrophic damage to these critical establishments. Screening gaps can arise due to the use of a less effective but faster screening method for high-risk passengers, which can happen as the combination of more passengers and limited resources can result in the unavailability of the ``right'' screening method. 

Online resource allocation is a problem of interest in many domains including transportation~\cite{simao2009approximate} -- allocating taxis to customers, emergency response~\cite{maxwell2010approximate} -- allocating ambulances to emergencies, and airports -- allocating terminals to arriving aeroplanes. Threat screening is also an online resource allocation problem, but in the presence of an observing adversary. Given that adversaries can monitor resource allocation strategies and exploit any gaps, we consider robust or risk-averse objectives rather than traditional expected objectives (e.g., expected revenue, expected delay). Hence, game-theoretic models and approaches have been considered for such problems. 

One such model is the Threat Screening Game (TSG) model introduced in \citet{brown2016one}, where a strategic attacker attempts to enter a secure area, while the screener uses teams of limited capacity screening resources with varying efficacy of catching the attacker to screen the screenees.
However, despite enhancements made in subsequent work~\cite{mccarthy2017staying}, this model suffers from a lack of adaptability as screening strategies are fixed for every hour. Further, all versions of the model assume a favourable rate of passenger arrival within each time window such that no screening resource is idle within the hour. We address these shortcomings with a novel online allocation model and a completely new solution approach.

Our \emph{first} contribution is an online version of the threat screening problem, in which the screening strategy is decided adaptively, based on the current queue lengths, as the screenees arrive. We show experimentally that this leads to a much better characterization and optimization of the average delay time faced by screenees at no loss to security risk (measured as attacker utility) compared to past work. Further, while past models have used a weight to balance the risk of missing an attacker and average delay time, we impose a hard bound on risk while simultaneously minimizing delay. We show that, given uncertain and unknown passenger arrivals, the online model can be solved as a Reinforcement Learning (RL) problem with continuous action space where the hard bound on risk translates to hard constraints on the action space. We show mathematically that the choice of hard bound on risk is not different from the one where a weighted defender objective is maximized and we can switch back and forth between these two seemingly different optimization goals. Our mathematical analysis also reveals the game-theoretic nature of this formulation.
%W

%We posit that this as a better choice as it allows screeners explicit control over the risk, ensuring that an adversary cannot take advantage of peak-hour rush to sneak through security (a genuine possibility when risk and average delay are combined together). 

Our \emph{second} contribution is a novel method to efficiently impose hard constraints on actions in Deep RL by using what we call $\alpha$-projection. In contrast to prior approaches (see Section \ref{sect:pastrl}), our approach guarantees that the constraint is never violated (even during training) while also being much more scalable in training as well as execution. The main component of our method is an extremely efficient mapping of infeasible actions to the feasible space specified by the constraints.

Finally, our \emph{third} contribution is a set of experiments that reveal why and how prior TSG models fail to handle realistic continuous arrival of passengers in bursts. The experiments also show that our approach achieves the same risk as prior models but improves upon the average delay by 100\% in the best case and 25\% on average.
% Moreover, our model can solve the largest instance of TSGs till the time of this writing, up to 300 flights per hour. We also show the superior performance of $\alpha$ projection against other baselines.
Overall, the realism of our model coupled with a novel scalable RL solution method makes our approach appealing for practical large scale threat screening problems.

\section{Related Work}
\subsection{TSG and Security Games}\label{sect:pasttsg}
There have been quite a few papers published on various aspects of threat screening games. The early papers~\cite{brown2016one,schlenker2017don} make two stringent assumptions---first, they assume perfect prior knowledge of passenger arrivals (for one hour time windows) and second, they implicitly assume that all passengers are screened within the same window in which they arrive, with no screening resource being idle (thus, delay is not an explicit consideration in this work). 
\includecomment{The second assumption could be clearer. There are actually two assumptions here: 1) Screening must be done in the same window and 2) Passengers arrive at favourable rates. Another point here is that even if they arrive at favourable rates, the assignment might not minimise delay.}
Both these assumptions are unrealistic in practice as, clearly, there is uncertainty in the number of passengers arriving in any time window and passengers arrive in bursts that are correlated with the flight timings. We show how these assumptions result in sub-optimal outcomes in practice. 
%\includecomment{The assumptions are incorrect, by why is that an issue? What happens if these assumptions do not hold/What constraints do they put on the solution? Delay cascades (Or do we talk about this elsewhere?)}

Later papers \cite{mccarthy2017staying,mccarthy2018price} attempt to account for uncertainty in arrivals across time windows and relax the second assumption by allowing an overflow of passengers from one time window to next. However, their approach is simply unscalable and hence impractical for real-world application. They show solutions for only up to 15 flights for a full day. Additionally, because the solution goal chosen is inspired from robust optimization, the method tries to find a solution that minimizes risk and delay across any realizable sample, which results in a pessimistic solution. The solution approach also makes the approximation of calculating the worst case from a sample of arrivals and, as a result, cannot guarantee that the solution will bound the true worst-case risk. Moreover, a problematic assumption from earlier work about the passengers arriving at a favourable rate within a time window continues to be assumed in this later work. In this paper, we propose a scalable online model without the restrictive assumptions of past work: we guarantee a bound on the worst-case risk and simultaneously minimize the average delay. Our online model allows for fine-grained adaptivity at the level of each passenger arrival as opposed to the hourly time-window adaptations in past work and also scales up to a large number of flights.

%However, the passengers arriving at a favourable rate within a time window so that no screening resource is idle is still assumed. 
%Further, the solution approach does not scale to real world problem sizes.  %\shortcite{mccarthy2017staying} attempts to find a static robust solution that minimises the worst-case risk across samples, which results in a pessimistic solution that is optimised for the worst case rather than the average case. 
%Finally,  

In other applications of security games played over multiple time steps, there has been work in which the defender strategy is a policy for an MDP or a sequence of actions~\cite{delle2014game,bosansky2015combining}. This work assumes that the MDP or game parameters are known beforehand and is hence a planning problem rather than a learning problem.\includecomment{what makes MDP parameters unknown in our problem?} Additionally, there has been work in this space where the double oracle approach has been used in tandem with Deep RL to compute the equilibrium~\cite{wang2019deep,wright2019iterated}. In this work, however, there is only a single constraint on actions (actions sum to one), which is readily enforced using a softmax layer. In this paper, our main contribution to Deep RL is in enforcing multiple arbitrary linear inequality constraints efficiently. There is also theoretical work on solving security games in an extensive form~\cite{letchford2010computing,kroer2018robust,vcerny2018incremental,basilico2009leader} or stochastic game~\cite{letchford2012computing} setting. Again, these assume complete knowledge of the game structure including transition functions whereas we focus on learning aspects. Also, whereas learning in the context of security games has appeared in the literature~\cite{balcan2015commitment,letchford2009learning}, this paper introduces an RL-based approach to threat screening for the first time.

\subsection{Constrained Action-Space RL}\label{sect:pastrl}
Historically, dealing with constraints on the action space in Deep RL has been a challenging task. This is exacerbated by the fact that our action space is continuous. The most common and intuitive technique is to discourage disallowed actions with a penalty. This method does not guarantee that the security risk will be bounded, however, and given that the risk is the worst-case allocation across an episode, the probability and extent of violation increases with scale. Given the adversarial nature of the security problem, this method is unsuitable.

Recently, \citet{pham2018optlayer} have suggested enforcing constraints by projecting any unconstrained point onto the constrained space by solving an optimisation program that minimises the L2 distance and back-propagating through it to train the network \citep{amos2017optnet}. This approach is very time consuming as it requires solving a quadratic program (QP) in the forward pass in every training iteration and, as a result, does not scale to problems with large dimensional action spaces~\cite{amos2017optnet} seen in practical screening problems. 
%Moreover, the method is also challenging to implement in practice. 

Our RL approach is similar in spirit to \citet{bhatia2019resource}, which uses a complicated variable-length iterative approximation of the L2 projection to deal with a specific subset of linear constraints faster than \citet{pham2018optlayer}. The type of linear constraints they can handle are constraints on the sum of sets of variables, where these sets must form a hierarchy. In contrast, the approach that we propose can handle arbitrary linear constraints, that is, in which the coefficients take any real values. Moreover, our approach is simple, can be computed in a single step, and is easy to implement as gradients can be computed using automated symbolic differentiation.

\section{MDP Model of TSG}\label{sect:MDP}
Our main departure from past TSG work is that we determine the screening strategy for a passenger when they arrive. Thus, while the model below reuses various notions from past versions of the TSG model, it models the online nature by formulating the problem as a Markov Decision Process (MDP). The MDP treatment considerably simplifies the problem of dealing with passenger arrival uncertainty by incorporating those in the MDP transition model. 

The basic structure of the screening problem stays the same as prior TSG models: every arriving passenger has a category $c \in C$, which is made up of two parts $\left< \theta, \kappa \right>$, where $\theta$ is the part of the category that the attacker cannot control (risk level determined by screener) and $\kappa$ is the part that they can control (which flight to take). The screening resource types comprise the set $R$; for example, $R$ could be \{X-Ray, Metal Detector, Advanced imaging\}. Each resource type $r \in R$ has a rate of screening (called capacity in prior work) given by $f_r$ measured in units of passengers per unit time. The passengers are screened by a team (set) of screening resources where the set of all teams $T$ is given a priori ($T \subset 2^R$). An attacker, apart from choosing a flight, uses an attack method $m$ (e.g., knife or gun) and each team $t$ has an effectiveness (probability) $E_{t,m}$ of detecting attack method $m$. $U_{\kappa,m}^{+}$ is the defender's utility for detecting an attacker with attacker's choice being $\kappa,m$ and $U_{\kappa,m}^{-}$ is the utility of not doing so.
% Need to think if below is needed
%We fix a reference of $U_{\kappa,m}^{+} = 0$ and hence $U_{\kappa,m}^{-} < 0$, as done in previous TSG models (note that affine transformations preserve the game outcome). 
As in previous TSG models, the adversary's utilities are the negation of these values. The defender has a belief about the attacker's uncontrollable category $\theta$ given by $P_\theta$ with $\sum_\theta P_\theta = 1$. 

\subsection{MDP Model}
Next, we describe our MDP formulation, which prescribes an online screening strategy for each arriving passenger. This is unlike past approaches in which the same randomized screening strategy was used for every passenger of a given category that arrived in the same time window.
%High level overview. What the constraints mean.

%Talk about how our approach also simplifies the problem of dealing with arrival uncertainty. Exact transitions unknown, based on data of past arrivals. 

%Define (S, A, T, R) tuple.

\begin{itemize}
    \item \textbf{States:} The state at any given point in time is a combination of 4 quantities $\left< c, \xi, h, \tau \right>$. The first, $c$, is the category of the passenger that has arrived for screening and we have to allocate security resources to. The remaining quantities summarise the history and provide information about the current context. $\xi \in \mathbb{R}^{|R|}$ encodes the number of passengers (or part thereof) in the queue for each resource at this current point in time. $h \in \mathbb{Z_+}^{|C|}$ is a summary of the history: it is the number of passengers from every category that have already been screened. $\tau$ is the wall clock time when the passenger arrives.
    \item \textbf{Actions:} An action $\pi_t \in\mathbb{R}^{|T|}$ at time step $t$ is a randomized allocation of the just-arrived passenger to teams. We use the insight that the risk is a function of the policy and does not depend on passenger arrivals to codify the hard bound as risk as constraints on the action space. Only actions with risk less than the specified risk level are allowed. Risk in TSGs is measured as the expected adversary utility, which is the negation of the utility of the defender. This leads to $m$ different inequalities constraints (one for each attack method) on the action stated in terms of defender utility.
    % Talk about how worst case risk translates to a hard bound.
    %We split these inequalities up into readable expressions below:
    \begin{align}
        P_{\theta} * \big[z_{m}U_{\kappa,m}^{+} + (1 - z_{m})U_{\kappa,m}^{-}\big] & \geq -\psi_\theta \quad \forall m , \label{eq:riskconstraint}\\
        \sum_{t \in T} E_{t, m}\pi_t  = z_m\quad \forall m, \mbox{ and }
        & \sum_{t \in T} \pi_t  = 1 \nonumber
    \end{align}
    % Talk about marginal policies and what they mean!
    Given a marginal policy $\pi$, $z_m$ is the overall probability that one of the teams $t \in T$ will detect an attack of type $m$. The last equality constraint enforces that $\pi_t$ is a probability distribution over teams. In order to explain the first set of inequalities, we first state the defender detection utility explicitly. $U_{\kappa,m} = z_{m}U_{\kappa,m}^{+} + (1 - z_{m})U_{\kappa,m}^{-}$ is the expected utility of the defender if the current passenger is an attacker. The utility $U_\theta = \min_{\kappa,m} U_{\kappa,m}$ is the worst case expected utility when the attacker is one with uncontrollable category $\theta$. $\sum_\theta P_\theta U_\theta$ is the overall defender detection expected utility. We wish to impose a lower bound $-\psi_\theta$ on $P_\theta U_\theta$ which indirectly lower bounds $\sum_\theta P_\theta U_\theta$ (or in other words, upper bounds risk). It is easy to see that the first set of inequalities is $P_\theta*( \min_m U_{\kappa,m}) \geq -\psi_\theta$. Since this inequality is applied for every passenger with uncontrollable category $\theta$, we get $P_\theta*( \min_{\kappa,m} U_{\kappa,m}) \geq -\psi_\theta$, which is nothing but $P_\theta U_\theta \geq -\psi_\theta$. Thus, this guarantees that the overall defender detection expected utility $\sum_\theta P_\theta U_\theta \geq - \sum_\theta \psi_\theta$ (or risk is bounded from above by $\sum_\theta \psi_\theta$). As these inequalities hold for any choice of action by attacker, this guarantees $- \sum_\theta \psi_\theta$ detection utility against a best responding attacker.

     %The first set of inequalities states the defender's detection utility for the current passenger %(which is always negative) 
    %must be above a minimum parametric threshold $-\psi_\theta$. Also, observe that 
    \item \textbf{Transitions:} The transition from $\left< c, \xi, h, \tau \right>$ to $\left< c', \xi', h', \tau' \right>$ can be decomposed into three parts. The first is how the allocation at the previous step and passage of wall clock time since then affects the queues for each resource. Passengers in the resources queues are screened according to the screening rate of a given resource:
    $$\xi'_r = \max(\xi_r - (\tau' - \tau)*f_r, 0) \mbox{ for all } r.$$
    The second part controls $h$:
    $$h'_c = h_c + 1, h'_d = h_c \mbox{ for all } d \neq c.$$
    The final part is determined by passenger arrivals and represents the likelihood of arrival of a passenger of a given type at a given time $P(c', \tau'| h, c, t)$. This is a function of the arrival history, but is unknown, which motivates our use of RL for the problem. 
    \item \textbf{Rewards:} The reward for each time step $t$ is the negative of the expected wait time of the currently arrived passenger. The wait time is determined by the maximum wait time over all resources in the realized team allocation. The wait time for each resource $r$ is determined by the screening rate $f_r$ and the number of passengers $\xi_r$ already in queue for that resource: $\xi_r/f_r$. We use $U_{o,t}$ to denote the delay reward at time $t$ (note $U_{o,t}$ is negative). The value (long term reward) is given by $V_o = \mathbb{E}[(1/N)\sum_{t=1}^N U_{o,t}]$, where $N$ passengers arrive in a day (implicitly conditional on the start state with empty history).
\end{itemize}

\subsection{Relationship to Game Theory}\label{sect:gametheory}
In the above constrained RL problem, the defender learns a policy which is a mixed strategy of the defender (mixed since the allocation at each time step is randomized). The adversary observes this policy and chooses an optimal attack $a$, which is a combination of $\kappa$ and attack method $m$. Thus, this is a Stackelberg game setting, similar to prior models of TSG. There are two components of the defender's value function: (a) the risk of not detecting the adversary, captured in $\sum_\theta P_\theta U_\theta $ and (b) the effect of delay, captured in $V_{o}$. 
The above RL approach solves the following problem $\max_{\pi \in \mathcal{F}_\psi} V_{o}(\pi)$ where $\pi$ represents policies and $\mathcal{F}_\psi = \{\pi ~\vert~ P_\theta U_{\theta}(\pi_t,a) \geq -\psi_\theta \mbox{ for all attacker actions $a$ and all } \theta\}$. Observe that here we explicitly write the arguments for $U_{\theta}$ and $V_o$. In particular, $V_o$ does not depend on the attacker action and the definition of $\mathcal{F}_\psi$ ensures achieving a minimum of $-\sum_\theta \psi_\theta$ detection utility against a \emph{best responding} adversary.

While the RL approach restricts the policy space of the defender via a bound on risk, one may wonder if the defender can achieve higher utility without such a restriction. Another way to view the problem is where the defender optimizes $\sum_\theta P_\theta U_\theta + w*V_o$ over all possible $\pi$ without any restrictions, where $w$ is a constant weight that specifies the relative importance of minimizing risk and average delay time of passengers. While our approach requires the defence agencies to specify acceptable risk level, this other approach requires specifying a trade-off weight $w$ between two completely different types of utilities (risk and delay), which is why we feel the hard bound on risk is more natural. But, in any case
%As we have argued before, defence agencies are themselves unsure of what weight $w$ to set, which is why we use our more practical approach of bounding risk. 
we show a relation between these two approaches that allows us to switch back and forth between them. 

\begin{theorem}
 There exists a $\psi$ (dependent on $w$) such that  any $\pi^* \in \arg\!\max_{\pi \in \mathcal{F}_\psi} V_{o}(\pi)$ is the defender strategy part of a Strong Stackelberg equilibrium of the Stackelberg game defined with defender objective as $\sum_\theta P_\theta U_\theta + w*V_o$.
\end{theorem}
\begin{proof}
An SSE is one which maximizes $\sum_\theta P_\theta U_\theta + w*V_o$, subject to the best response of the attacker. The definition $U_\theta = \min_{\kappa,m} U_{\kappa,m}$ already takes care of the best response of the attacker as the attacker utility is $-U_{\kappa,m}$ and the attacker action $\kappa,m$ that  minimizes $U_{\kappa,m}$ maximizes $-U_{\kappa,m}$. Also, as $U_{\kappa,m}$ is continuous in $\pi$ and min of continuous functions is continuous, $\sum_\theta P_\theta U_\theta$ is continuous in $\pi$.

The space of possible $\pi$ is compact, thus, the continuous bounded function $\sum_\theta P_\theta U_\theta + w*V_o$ of $\pi$ achieves a maximum at some $\pi^* \in \Pi^*$. This set $\Pi^*$ is the set of defender strategies that form a SSE of the game. Let the value $U^*_\theta$ and $V_o^*$ be obtained at this $\pi^*$. 

Consider the values $-\psi^*_\theta = P_\theta U^*_\theta$ and $\psi^* = \langle \psi^*_\theta \rangle_{\theta \in \Theta}$.  Note that  $\mathcal{F}_\psi$ is specified by linear inequalities given by Equation~\ref{eq:riskconstraint}, thus, $\mathcal{F}_\psi$ is a polytope.  Also, $\pi^* \in \mathcal{F}_{\psi^*}$. 
We claim that the optimal solution of $\max_{\pi \in \mathcal{F}_{\psi^*}} V_{o}(\pi)$ is in $\Pi^* \cap \mathcal{F}_{\psi^*}$, which is not empty as $\pi^* \in \Pi^* \cap \mathcal{F}_{\psi^*}$. 
As $\sum_\theta P_\theta U_\theta^* + w*V_o^*$ is the global maximum, for any $\pi \in \mathcal{F}_{\psi^*}$ if $\sum_\theta P_\theta U_\theta + V_o  = \sum_\theta P_\theta U^*_\theta + V_o^*$ then $\pi \in \Pi^*$. And also, there does not exist any $\pi \in \mathcal{F}_{\psi^*}$ such that $\sum_\theta P_\theta U_\theta + V_o  > \sum_\theta P_\theta U^*_\theta + V_o^*$, which proves our claim.
Since the optimal solution is in $\Pi^*$, this proves our result for $\psi^*$.
\end{proof}

The above theorem also provides an easy algorithm to solve for an approximate SSE in the unrestricted game using the RL approach. The approach is to construct a Pareto frontier a priori by solving for the optimal policy for many values of $\psi$ where these values are uniformly spaced and distributed throughout the possible space of $\psi$ values. Then, when given $w$, the solution will choose one of the specific points for which the output $\pi$ maximizes $\sum_\theta P_\theta U_\theta + w*V_o$ over all points considered in the $\psi$ space.

%We reformulate it in this way, but what game does this new solution correspond to? The best delay for a given risk level lies on the pareto frontier and corresponds to an optimal solution for some combined objective of risk and delay. This joint objective, in turn, can be viewed as the equilibrium point of some sort of pseudo-zero sum game(?).

\section{Overall Solution}
To solve the screening problem modelled in Section \ref{sect:MDP}, we use Reinforcement Learning (RL). We use techniques from RL instead of trying to solve the MDP directly because the exact passenger arrival distribution is unknown. Rather than trying to model the distribution explicitly, we use model-free RL techniques to jointly learn the distribution and the optimal policy.

Specifically, we use the Deep Deterministic Policy Gradient (DDPG)~\citep{lillicrap2015continuous} algorithm that is a state-of-the-art technique in Deep Reinforcement Learning literature. DDPG is an extension of the standard actor-critic approach that allows the modelling of a continuous action space like the one present in our problem. However, DDPG cannot enforce general action-space constraints as is. 

To deal with this, we propose an $\alpha$-projection layer in Section \ref{sect:alphaproj} that enforces constraints on the output of the previous layer. We then modify the standard DDPG algorithm by adding this $\alpha$-projection layer on top of the output layer of the actor network. This ensures that any action produced by the actor satisfies the risk constraints on the action space. This combination of DDPG and $\alpha$-projection represents the overall approach that what we use to solve the MDP.

\section{Approach for Linear Constraints on Action Space in Deep RL}
%(Have a diagram of an L2 mapping versus our mapping.) - yes good to have

All prior approaches for imposing \emph{hard} constraints on the action output of any policy neural network use a layer(s) at the end of the network to map the unconstrained output from intermediate layers to an output in the feasible space. Mathematically, suppose the output must lie in a feasible space $Y$ defined by linear inequality constraints. Let $f(x)$ be the output of the intermediate layer, given input $x$. The last layer(s) define a mapping $M$ such that $M(f(x))$ lies in the feasible space $Y$. For our problem, $Y$ is a fixed polytope for a given problem instance (see Equation~\ref{eq:riskconstraint}). Typically, such mappings $M$ have been some type of $L_p$ projection ($p = 1 \text{ or } 2$) in the past. This projection is written as an optimization problem and enforced as a neural network layer using techniques such as OptLayer~\cite{pham2018optlayer}. However, such mappings are expensive to compute in practice as they require solving a quadratic program for every training iteration and every execution. This creates the case for a simpler mapping. First, we list desirable properties of such a mapping.

%OptLayer idea: Do projection to minimise L2 Loss. Key insight: Exact mapping of unconstrained to constrained is unimportant unimportant because the neural network will learn through it. The choice of mapping only needs to be:
\begin{itemize}
    \item Onto: To make sure that the neural network has the opportunity to output any value in the feasible space. This ensures that there is no loss in solution quality arising from a restricted solution space.
    \item Continuous everywhere and differentiable almost everywhere: This allows the neural network to learn through this mapping layer by backpropagating gradients.\footnote{Non-differentiability for points in a measure zero set are allowed, as is in the ReLU activation.} \item No vanishing or exploding gradients: The gradients of the mapping should be informative when optimizing loss, that is, should not be zero or very large for many points. 
\end{itemize}
Observe that the mapping $M$ need not be a closest point (in any distance) in $Y$ to the unconstrained $f(x)$. This is because the whole neural network is the function composition $M \circ f$ and the training ensures that the output $M \circ f$ minimizes the loss. The purpose of $M$ is to only ensure feasibility of output, thus, any choice of $M$ (with the properties above) will work since the neural network will appropriately adjust $f$ so that $M \circ f$ is optimal.
%Talk about our mapping what our mapping tries to do at a high level. Talk about why it's a lot easier (1 variable) - faster and easier to implement (can use off-the shelf symbolic differentiation).

\begin{figure}
\tikzset{every picture/.style={line width=0.75pt}} %set default line width to 0.75pt        
\centering
\begin{tikzpicture}[x=0.75pt,y=0.75pt,yscale=-0.45,xscale=0.45]
%uncomment if require: \path (0,278.6666564941406); %set diagram left start at 0, and has height of 278.6666564941406

%Shape: Ellipse [id:dp8628090192847597] 
\draw  [color={rgb, 255:red, 126; green, 211; blue, 33 }  ,draw opacity=1 ][fill={rgb, 255:red, 126; green, 211; blue, 33 }  ,fill opacity=1 ] (551.35,159.5) .. controls (551.23,165.13) and (546.33,169.61) .. (540.41,169.49) .. controls (534.49,169.38) and (529.78,164.72) .. (529.9,159.09) .. controls (530.01,153.45) and (534.91,148.98) .. (540.83,149.09) .. controls (546.76,149.21) and (551.46,153.86) .. (551.35,159.5) -- cycle ;
%Shape: Ellipse [id:dp23491265280473317] 
\draw  [color={rgb, 255:red, 208; green, 2; blue, 27 }  ,draw opacity=1 ][fill={rgb, 255:red, 208; green, 2; blue, 27 }  ,fill opacity=1 ] (318.18,218.47) .. controls (318.06,224.1) and (313.17,228.58) .. (307.24,228.46) .. controls (301.32,228.35) and (296.61,223.69) .. (296.73,218.06) .. controls (296.85,212.42) and (301.75,207.95) .. (307.67,208.06) .. controls (313.59,208.18) and (318.3,212.83) .. (318.18,218.47) -- cycle ;
%Shape: Ellipse [id:dp5800612347727998] 
\draw  [color={rgb, 255:red, 208; green, 2; blue, 27 }  ,draw opacity=1 ][fill={rgb, 255:red, 208; green, 2; blue, 27 }  ,fill opacity=1 ] (235.69,262.62) .. controls (235.57,268.25) and (230.68,272.73) .. (224.75,272.61) .. controls (218.83,272.5) and (214.12,267.84) .. (214.24,262.21) .. controls (214.36,256.58) and (219.25,252.1) .. (225.18,252.21) .. controls (231.1,252.33) and (235.81,256.99) .. (235.69,262.62) -- cycle ;
%Shape: Ellipse [id:dp7785272079302354] 
\draw  [color={rgb, 255:red, 208; green, 2; blue, 27 }  ,draw opacity=1 ][fill={rgb, 255:red, 208; green, 2; blue, 27 }  ,fill opacity=1 ] (64.37,262.77) .. controls (64.25,268.41) and (59.36,272.88) .. (53.43,272.77) .. controls (47.51,272.65) and (42.8,267.99) .. (42.92,262.36) .. controls (43.04,256.73) and (47.93,252.25) .. (53.86,252.37) .. controls (59.78,252.48) and (64.49,257.14) .. (64.37,262.77) -- cycle ;
%Shape: Ellipse [id:dp42851710778156105] 
\draw  [color={rgb, 255:red, 208; green, 2; blue, 27 }  ,draw opacity=1 ][fill={rgb, 255:red, 208; green, 2; blue, 27 }  ,fill opacity=1 ] (145.71,49.4) .. controls (145.59,55.03) and (140.69,59.5) .. (134.77,59.39) .. controls (128.85,59.28) and (124.14,54.62) .. (124.26,48.98) .. controls (124.37,43.35) and (129.27,38.88) .. (135.19,38.99) .. controls (141.12,39.1) and (145.82,43.76) .. (145.71,49.4) -- cycle ;
%Shape: Polygon [id:ds01339970913774291] 
\draw   (138.48,275.43) -- (0.53,183.62) -- (50.73,10.8) -- (263.37,41.17) -- (329.38,173.9) -- cycle ;
%Shape: Ellipse [id:dp3368179778370579] 
\draw  [color={rgb, 255:red, 208; green, 2; blue, 27 }  ,draw opacity=1 ][fill={rgb, 255:red, 208; green, 2; blue, 27 }  ,fill opacity=1 ] (337.19,88.5) .. controls (337.07,94.13) and (332.17,98.61) .. (326.25,98.49) .. controls (320.33,98.38) and (315.62,93.72) .. (315.74,88.09) .. controls (315.86,82.46) and (320.75,77.98) .. (326.68,78.09) .. controls (332.6,78.21) and (337.31,82.87) .. (337.19,88.5) -- cycle ;
%Straight Lines [id:da7958760780071039] 
\draw    (326.46,88.29) -- (292.2,97.93) ;
\draw [shift={(292.2,97.93)}, rotate = 209.29] [color={rgb, 255:red, 0; green, 0; blue, 0 }  ][line width=0.75]    (-7.59,0) -- (7.59,0)(0,7.59) -- (0,-7.59)   ;

%Straight Lines [id:da1403297434403039] 
\draw  [dash pattern={on 0.84pt off 2.51pt}]  (326.46,88.29) -- (158.11,139.04) ;

%Straight Lines [id:da9712661098314868] 
\draw  [dash pattern={on 0.84pt off 2.51pt}]  (307.46,218.26) -- (158.11,139.04) ;

%Straight Lines [id:da5041957398621295] 
\draw  [dash pattern={on 0.84pt off 2.51pt}]  (224.96,262.41) -- (158.11,139.04) ;

%Straight Lines [id:da6858660609961569] 
\draw  [dash pattern={on 0.84pt off 2.51pt}]  (53.65,262.57) -- (158.11,139.04) ;

%Straight Lines [id:da9757917373103324] 
\draw  [dash pattern={on 0.84pt off 2.51pt}]  (11.21,60.54) -- (158.11,139.04) ;

%Straight Lines [id:da9791281235579132] 
\draw  [dash pattern={on 0.84pt off 2.51pt}]  (199.88,10.42) -- (158.11,139.04) ;

%Straight Lines [id:da5607205863993341] 
\draw    (307.46,218.26) -- (277.24,201.68) ;
\draw [shift={(277.24,201.68)}, rotate = 270] [color={rgb, 255:red, 0; green, 0; blue, 0 }  ][line width=0.75]    (-7.59,0) -- (7.59,0)(0,7.59) -- (0,-7.59)   ;

%Straight Lines [id:da9940613969213843] 
\draw    (224.96,262.41) -- (210.96,236.35) ;
\draw [shift={(210.96,236.35)}, rotate = 286.76] [color={rgb, 255:red, 0; green, 0; blue, 0 }  ][line width=0.75]    (-7.59,0) -- (7.59,0)(0,7.59) -- (0,-7.59)   ;

%Straight Lines [id:da9143432717641262] 
\draw    (199.88,10.42) -- (193.35,31.82) ;
\draw [shift={(193.35,31.82)}, rotate = 151.97] [color={rgb, 255:red, 0; green, 0; blue, 0 }  ][line width=0.75]    (-7.59,0) -- (7.59,0)(0,7.59) -- (0,-7.59)   ;

%Straight Lines [id:da5069133345320593] 
\draw    (53.65,262.57) -- (77.28,233.87) ;
\draw [shift={(77.28,233.87)}, rotate = 354.48] [color={rgb, 255:red, 0; green, 0; blue, 0 }  ][line width=0.75]    (-7.59,0) -- (7.59,0)(0,7.59) -- (0,-7.59)   ;

%Straight Lines [id:da858695042092267] 
\draw    (11.21,60.54) -- (32.86,71.63) ;
\draw [shift={(32.86,71.63)}, rotate = 72.11] [color={rgb, 255:red, 0; green, 0; blue, 0 }  ][line width=0.75]    (-7.59,0) -- (7.59,0)(0,7.59) -- (0,-7.59)   ;

%Straight Lines [id:da49226751186299245] 
\draw    (307.46,218.26) -- (293.03,194.52) ;
\draw [shift={(291.99,192.81)}, rotate = 418.7] [color={rgb, 255:red, 0; green, 0; blue, 0 }  ][line width=0.75]    (10.93,-3.29) .. controls (6.95,-1.4) and (3.31,-0.3) .. (0,0) .. controls (3.31,0.3) and (6.95,1.4) .. (10.93,3.29)   ;

%Straight Lines [id:da7654952041476035] 
\draw  [dash pattern={on 0.84pt off 2.51pt}]  (134.98,49.19) -- (158.11,139.04) ;

%Straight Lines [id:da4618801638231931] 
\draw    (134.98,49.19) ;
\draw [shift={(134.98,49.19)}, rotate = 45] [color={rgb, 255:red, 0; green, 0; blue, 0 }  ][line width=0.75]    (-7.59,0) -- (7.59,0)(0,7.59) -- (0,-7.59)   ;

%Straight Lines [id:da19641020704817325] 
\draw    (580.96,221.24) -- (492.28,185.62) ;
\draw [shift={(492.28,185.62)}, rotate = 270.0] [color={rgb, 255:red, 0; green, 0; blue, 0 }  ][line width=0.75]    (-9.59,0) -- (9.59,0)(0,9.59) -- (0,-9.59)   ;

%Straight Lines [id:da2783180899679829] 
\draw    (579.05,138.7) -- (518.24,26.37) ;

%Straight Lines [id:da3677616555341221] 
\draw    (579.05,138.7) -- (422.45,223.62) ;

%Shape: Ellipse [id:dp5931609171065633] 
\draw  [dash pattern={on 4.5pt off 4.5pt}] (354.77,186.21) .. controls (354.13,216.82) and (328.32,241.14) .. (297.11,240.54) .. controls (265.91,239.94) and (241.12,214.65) .. (241.76,184.04) .. controls (242.4,153.43) and (268.21,129.1) .. (299.42,129.7) .. controls (330.62,130.3) and (355.4,155.6) .. (354.77,186.21) -- cycle ;
%Shape: Ellipse [id:dp25931286856891367] 
\draw  [dash pattern={on 4.5pt off 4.5pt}] (650.47,151.57) .. controls (649.18,220.06) and (591.55,274.58) .. (521.74,273.37) .. controls (451.93,272.15) and (396.39,215.64) .. (397.68,147.16) .. controls (398.97,78.68) and (456.61,24.15) .. (526.41,25.37) .. controls (596.22,26.59) and (651.76,83.09) .. (650.47,151.57) -- cycle ;
%Straight Lines [id:da7133303882329072] 
\draw    (579.8,221.22) -- (541.69,160.98) ;
\draw [shift={(540.62,159.29)}, rotate = 417.68] [color={rgb, 255:red, 0; green, 0; blue, 0 }  ][line width=0.75]    (10.93,-3.29) .. controls (6.95,-1.4) and (3.31,-0.3) .. (0,0) .. controls (3.31,0.3) and (6.95,1.4) .. (10.93,3.29)   ;

%Straight Lines [id:da5908316945667149] 
\draw  [dash  pattern={on 0.84pt off 2.51pt}]  (580.96,221.24) -- (397.68,147.16) ;

%Straight Lines [id:da7771161677350669] 
\draw  [dash pattern={on 0.84pt off 2.51pt}]  (621.33,181.93) -- (398.83,116.89) ;

%Straight Lines [id:da857701784168101] 
\draw [dash pattern={on 4.5pt off 4.5pt}]  (356.29,194.67) -- (398.05,184.15) ;
\draw [shift={(399.99,183.67)}, rotate = 525.87] [color={rgb, 255:red, 0; green, 0; blue, 0 }  ][line width=0.75]    (10.93,-3.29) .. controls (6.95,-1.4) and (3.31,-0.3) .. (0,0) .. controls (3.31,0.3) and (6.95,1.4) .. (10.93,3.29)   ;

%Shape: Ellipse [id:dp7638497249028156] 
\draw  [color={rgb, 255:red, 208; green, 2; blue, 27 }  ,draw opacity=1 ][fill={rgb, 255:red, 208; green, 2; blue, 27 }  ,fill opacity=1 ] (210.61,10.62) .. controls (210.49,16.26) and (205.59,20.73) .. (199.67,20.62) .. controls (193.75,20.5) and (189.04,15.85) .. (189.16,10.21) .. controls (189.27,4.58) and (194.17,0.1) .. (200.09,0.22) .. controls (206.02,0.33) and (210.72,4.99) .. (210.61,10.62) -- cycle ;
%Shape: Ellipse [id:dp45350908189717276] 
\draw  [color={rgb, 255:red, 208; green, 2; blue, 27 }  ,draw opacity=1 ][fill={rgb, 255:red, 208; green, 2; blue, 27 }  ,fill opacity=1 ] (21.94,60.75) .. controls (21.82,66.38) and (16.92,70.86) .. (11,70.74) .. controls (5.08,70.63) and (0.37,65.97) .. (0.49,60.34) .. controls (0.6,54.7) and (5.5,50.23) .. (11.42,50.34) .. controls (17.35,50.46) and (22.05,55.11) .. (21.94,60.75) -- cycle ;
%Shape: Ellipse [id:dp1639592513335697] 
\draw  [color={rgb, 255:red, 0; green, 0; blue, 0 }  ,draw opacity=1 ][fill={rgb, 255:red, 0; green, 0; blue, 0 }  ,fill opacity=1 ] (168.83,139.25) .. controls (168.72,144.88) and (163.82,149.36) .. (157.9,149.24) .. controls (151.97,149.13) and (147.27,144.47) .. (147.38,138.84) .. controls (147.5,133.2) and (152.4,128.73) .. (158.32,128.84) .. controls (164.24,128.96) and (168.95,133.61) .. (168.83,139.25) -- cycle ;

% Text Node
\draw (95,145) node [rotate=-0.0] [align=left] {Center};
% Text Node
\draw (650,155) node [fill=white,rotate=-0.0] [align=left] {P1'};
% Text Node
\draw (558,250) node [fill=white,rotate=-0.0] [align=left] {P1};
% Text Node
\draw (285,250) node [fill=white,rotate=-0.0] [align=left] {P1};
% Text Node
\draw (530.28,139.84) node [outer sep=2pt,fill=white,above,rotate=-0.0] [align=left] {Optimal};

%Shape: Ellipse [id:dp25786559899220185] 
\draw  [color={rgb, 255:red, 208; green, 2; blue, 27 }  ,draw opacity=1 ][fill={rgb, 255:red, 208; green, 2; blue, 27 }  ,fill opacity=1 ] (632.05,182.14) .. controls (631.93,187.77) and (627.04,192.24) .. (621.11,192.13) .. controls (615.19,192.02) and (610.48,187.36) .. (610.6,181.73) .. controls (610.72,176.09) and (615.61,171.62) .. (621.54,171.73) .. controls (627.46,171.84) and (632.17,176.5) .. (632.05,182.14) -- cycle ;
%Shape: Ellipse [id:dp1795020104036542] 
\draw  [color={rgb, 255:red, 208; green, 2; blue, 27 }  ,draw opacity=1 ][fill={rgb, 255:red, 208; green, 2; blue, 27 }  ,fill opacity=1 ] (591.69,221.45) .. controls (591.57,227.08) and (586.67,231.55) .. (580.75,231.44) .. controls (574.83,231.33) and (570.12,226.67) .. (570.24,221.03) .. controls (570.35,215.4) and (575.25,210.93) .. (581.17,211.04) .. controls (587.1,211.15) and (591.8,215.81) .. (591.69,221.45) -- cycle ;
\end{tikzpicture}
\caption{\textbf{Left}: The red points are produced by an intermediate layer. The $\alpha$-projection is shown with cross marks. If the red point is inside the feasible region, the $\alpha$-projection of the point is the point itself. For one red point (inside the small dashed circle), the L2 projection is shown with the arrow. \textbf{Right}: A zoomed-in version of the dashed circle. This shows that if the L2 projection of P1 is the optimal point (green), then the intermediate layer will adjust its output to P1' to get the optimal point with the $\alpha$-projection.} \label{fig:alphaexplain}
\end{figure}
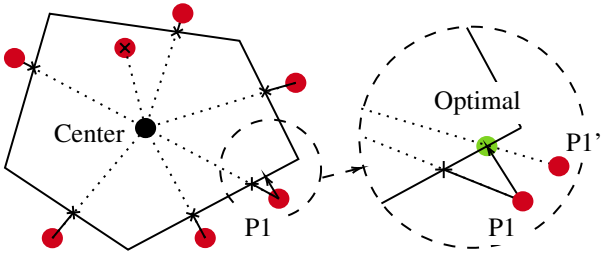

% Text Node
%\draw (177,298.33) node  [align=left] {\small The red points are produced by intermediate \\layer. $\alpha$ projection is shown with the cross marks.\\If red point is inside, the alpha projection is itself.\\For one red point (left corner), the L2 projection \\is shown with the arrow.};
% Text Node
%\draw (508,298.33) node  [align=left] {A zoomed version of the dashed circle.\\This shows that if the L2 projection of P1 is \\the optimal point (green) then the intermediate \\layer will adjust its output to P2 to get the \\optimal point with the alpha projection.};
% Text Node

\subsection{$\alpha$-projection}\label{sect:alphaproj}
We propose a very simple mapping $M$, which we call $\alpha$-projection. We start by finding a feasible point $y_0 \in int(Y)$, where $int(Y)$ are all the interior points of $Y$. Then, we find the maximum $\alpha \in [0,1]$ such that $y = \alpha *f(x) + (1-\alpha)*y_0$ and $y \in Y$, with the final output being the point $y$, so $M(f(x)) = y$. Intuitively, we join $f(x)$ and $y_0$ with a line and choose the closest point to $f(x)$ on this line that lies in $Y$, which could be $f(x)$ itself if $f(x) \in Y$. See Figure~\ref{fig:alphaexplain}.

\textbf{Forward pass} $\alpha$-projection turns out to be very efficient for linear constraints such as those in the TSG-RL problem. For training neural networks, every iteration requires a forward pass for the network also. An optimization layer (such as our $\alpha$-projection layer) requires solving an optimization problem. However, our simple mapping allows obtaining $\alpha$ in a closed-form with no need to solve expensive optimizations. To get to the closed-form, first, lets consider $m$ inequality constraint with the $i^{th}$ constraint being $a_i \cdot y \leq b_i$. Using our mapping, this $i^{th}$ constraint is $ \alpha *( a_i \cdot f(x)) +  (1-\alpha)*(a_i \cdot y_0) \leq b_i$. Thus, $\alpha * (a_i \cdot f(x) - a_i \cdot y_0) \leq b - a_i \cdot y_0 $. We get a similar upper (or lower, depending on the sign of $a_i \cdot f(x) - a_i \cdot y_0$) bound for $\alpha$ for every inequality constraint and then $\alpha$ is the simply the highest value in $[0,1]$ that satisfies all these bounds. All these bounds are closed-form formulas, thus, computing $\alpha$ is very efficient for the forward pass, and obtaining $y = \alpha*f(x) + (1-\alpha)*y_0$ is easy.

\textbf{Gradients}: Computing the gradient (for backpropagation) of $\alpha$ w.r.t. the input $f(x)$ to the mapping layer is also easy. For readability, we write $s$ instead of $f(x)$. As stated in the previous paragraph $\alpha$ is the minimum of a number of upper bounds (ignroing lower bounds), where these upper bounds are given closed-form functions $b_i(s)$ ($y_0$ is a constant). Thus, $\alpha = \min(1,b_1(s), \ldots, b_k(s))$ for some $k < m$. For notational ease, let $b_0(s) = 1$. For any specific $s_0$, if there is a unique index $j$ for which $\alpha = b_j(s)$ then the gradient is simply $\nabla b_j(s)$, which is $0$ if $j=0$. If there is a set of indices $J$ ($|J|>1$) and $\alpha = b_j(s)$ for all $j \in J$, then the gradient is simply $(1/|J|)*(\sum_{j\in J}\nabla b_j(s))$. In practice, these gradients need not be explicitly calculated and can be handled by automatic symbolic differentiation libraries~\cite{tensorflow2015-whitepaper} instead. Thus, our approach is also simple to implement.

\textbf{Handling equality constraints}: For equality constraints, say $k$ constraints, the general approach would be to eliminate $k$ variables using Gaussian elimination (or any other method) and then deal with the only inequalities in this new space. However, for our TSG problem, we only have one equality constraint, which is a probability simplex constraint that can be easily enforced by a softmax layer. With a slight abuse of notation, we use $s$ to denote the output of the softmax whose input is the unconstrained output $f(x)$. $s$ satisfies the probability simplex constraint. Additionally, we choose $y_0$ such that it also satisfies the probability simplex constraint. Then, the final output $\alpha*s + (1-\alpha)*y_0$ satisfies the probability simplex constraint and all the inequality constraints. Thus, for our problem, the overall mapping is made of 2 layers: a softmax followed by the $\alpha$-projection layer.

%\subsection{Closed-form solution to optimisation problem above(?)} 
%For linear constraints, some interior point c and some active constraints A, b: write the equation for alpha. 

% lets leave this for the more general paper
%\subsection{Constraints as a function of the input}
%Constraints are not constant. Physical significance of a constraint and about why different points need different constraints.

\subsection{Choosing an Interior Point}
The choice of $y_0$ is important. First, $y_0$ should be an interior point. Otherwise, if $y_0$ is on an external face (hyperplane) of the $Y$ polytope, all external points on the side of the face not containing the polytope will map to $y_0$. This violates non-zero gradients property of a feasible mapping, as all points on one side of the hyperplane will have zero gradients. We also find that we get better performance when $y_0$ is near the centre of the polytope. While there are many different types of centres, we choose the Chebyshev centre of a polytope because: (1) Chebyshev centre can be computed efficiently by solving a linear program and (b) the Chebyshev centre maximizes the minimum distance from the faces of the polytope, that is, making sure the centre is far from bad points. Informally, Chebyshev centre is the centre of the largest ball that fits inside the polytope. Note that we need to compute the Chebyshev centre only once for our polytope $Y$ and that this computation time is of the order of seconds.

%What are the considerations? Any interior point. Cannot lie on the boundary (intuition). Additionally, it would be nice to have a central point. Reason. 
%We use Chebyshev centre because it is easy to calculate. But centre of maximum volume ellipsoid or some other unique centre can also be used. 

%Consistent interior point, in our case. What the physical significance of an interior point is and why there will always be a consistent interior point for all inequalities.

%\subsection{Dealing with sum-to-one constraint}
%Input y is the output of a softmax. Follows the rule. While choosing interior point c, enforce this constraint. Then any point on the line between the two will also sum-to-one.

\section{Experiments}
In line with past work on TSGs, we evaluate the performance of our approach on the airport passenger screening domain. 
%The authors of past work have not released the code associated with their solutions, however. 
%To account for any implementation specific differences, 
For the most simple head-to-head comparison, we look at the difference in solution quality between our approach and past work within single time-window.  \citet{brown2016one} and \citet{mccarthy2017staying} both have the same optimal solution in this case and the optimal marginal solution can be found by using a simple linear program (LP). When we compare the solutions of the LP to our approach, we control for the risk and measure the corresponding difference in delay. Specifically, we take the different risk levels associated with the uncontrollable categories $\psi_\theta$ from the solution of the LP and run our approach using those as the risk threshold in the risk constraints of our approach. We then test both sets of policies (LP and ours) using an online simulator and compare the ratio of average delays obtained. 

We construct our problem instances using the description in~\citet{brown2016one} and~\citet{mccarthy2017staying}. The attacker utility associated with successfully launching an attack $U^{+}$ is sampled from a uniform distribution over [1, 10], while the utility of failing to launch an attack $U^-$ is set to 0. The game is zero-sum and, as a result, the defender utilities are the negation of the attacker utilities. There are 3 attack methods $m$, 5 uncontrollable screenee risk levels $\theta$ and 5 screening resource types $|R|$. The efficacies (probability of detection) of different resources are sampled from a uniform random distribution over [0, 1] for each attack method. We create 10 random 2-sized combinations of resources to represent the teams $T$ and their efficacies are found assuming that the efficacies of their associated resources are independent. We choose a passenger arrival distribution as used in \citet{mccarthy2017staying}, and consider arrivals to be normally distributed in a 3-hour window leading up to the departure of the flight. We combine this with real flight departure times taken from one of the busiest airports in the world to generate a realistic arrival distribution of passengers. The default number of flight types for experiments is 10.

Finally, for runtime, all past methods have used non-gradient based optimization methods and have reported runtimes for programs that have run on CPUs. As has been observed in other domains, GPUs offer a huge advantage due to the immense parallelization of matrix operations for neural network optimization. However, to perform a fair comparison to past work, we run all our experiments on a CPU. Thus, while our scalability results show the runtime trend with increasing problem size, the absolute wall clock time can be much better with GPUs.

\subsection{Static vs. Adaptive}
\begin{figure}
    \centering
    \includegraphics[ width=\hsize]{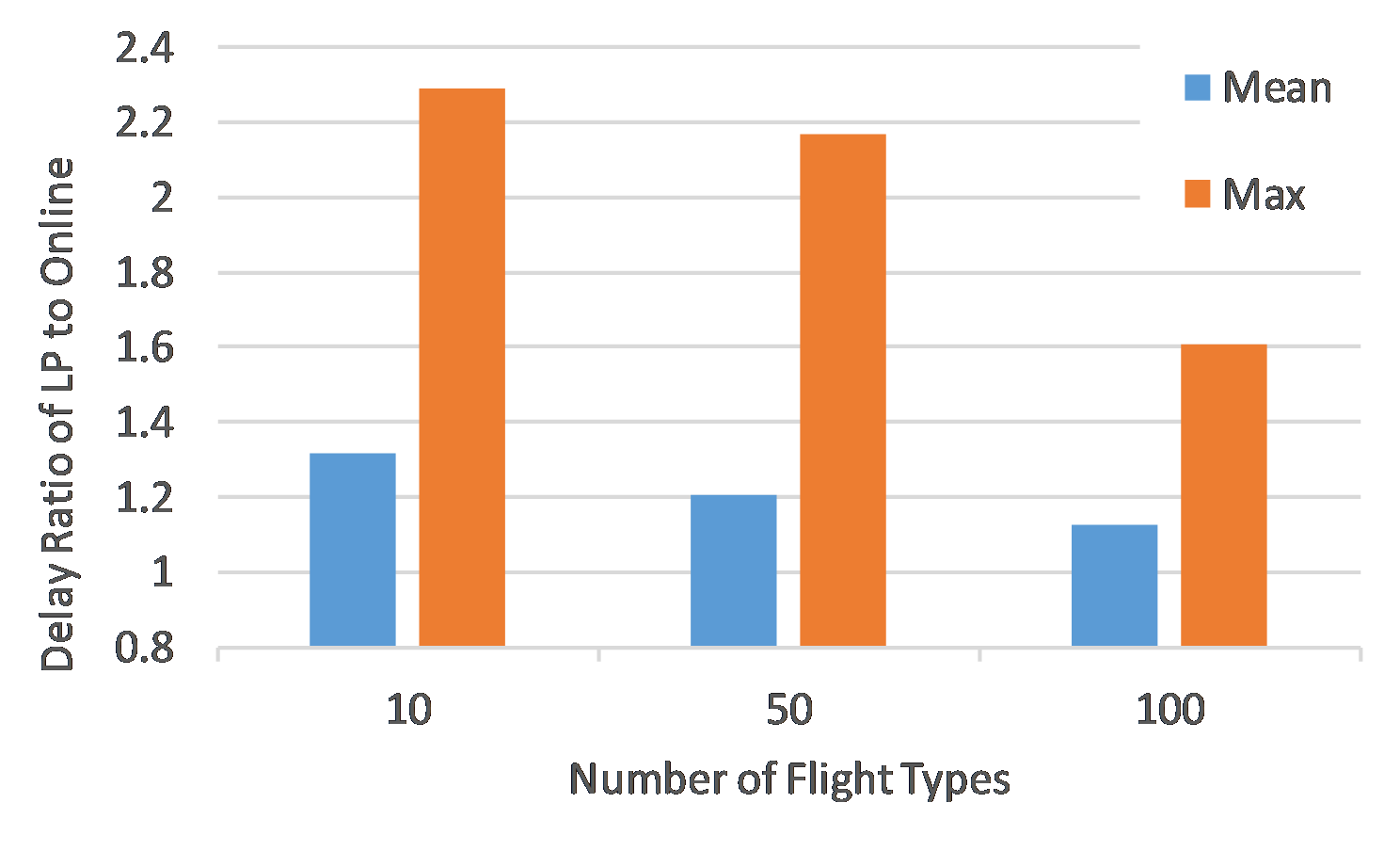}
    \caption{\textbf{Static vs. Adaptive:} We measure the ratio of average delays of the LP to average delays of our online approach (so higher is better) across different problem sizes.}
    \label{fig:staticvsadaptive}
\end{figure}
First, we show that our online approach outperforms past window-based approaches even when the problem size is large. The experiments here are evaluated for 30 random game instances and averaged across 100 samples of passenger arrival sequence. The results can be seen in Figure \ref{fig:staticvsadaptive}. 

For the one time-window problem, improvement in solution quality comes from the fact that past work has a static policy within one time-window, whereas our solution can adapt based on the actual number of passenger arrivals. As a result, our approach can exploit the structure present within a time-window. The reason our improvement decreases with an increase in the number of flights is because, as the problem size increases, the structure present in a randomly generated problem decreases. For example, with many overlapping Gaussian distribution of passenger arrivals, the overall arrival is almost uniform which we show leads to a lower gain (see Section~\ref{sec:var}).

Moreover, the performance improvement within one time-window is a lower bound on the amount of improvement we can achieve over past methods. The solution quality associated with past methods deteriorates with an increase in the number of windows (details in Section \ref{sect:pasttsg}) while there is no notion of time windows in our model and hence no degradation over long periods.

\subsection{Delay vs. Risk}
\begin{figure}
    \centering
    \includegraphics[ width=\hsize]{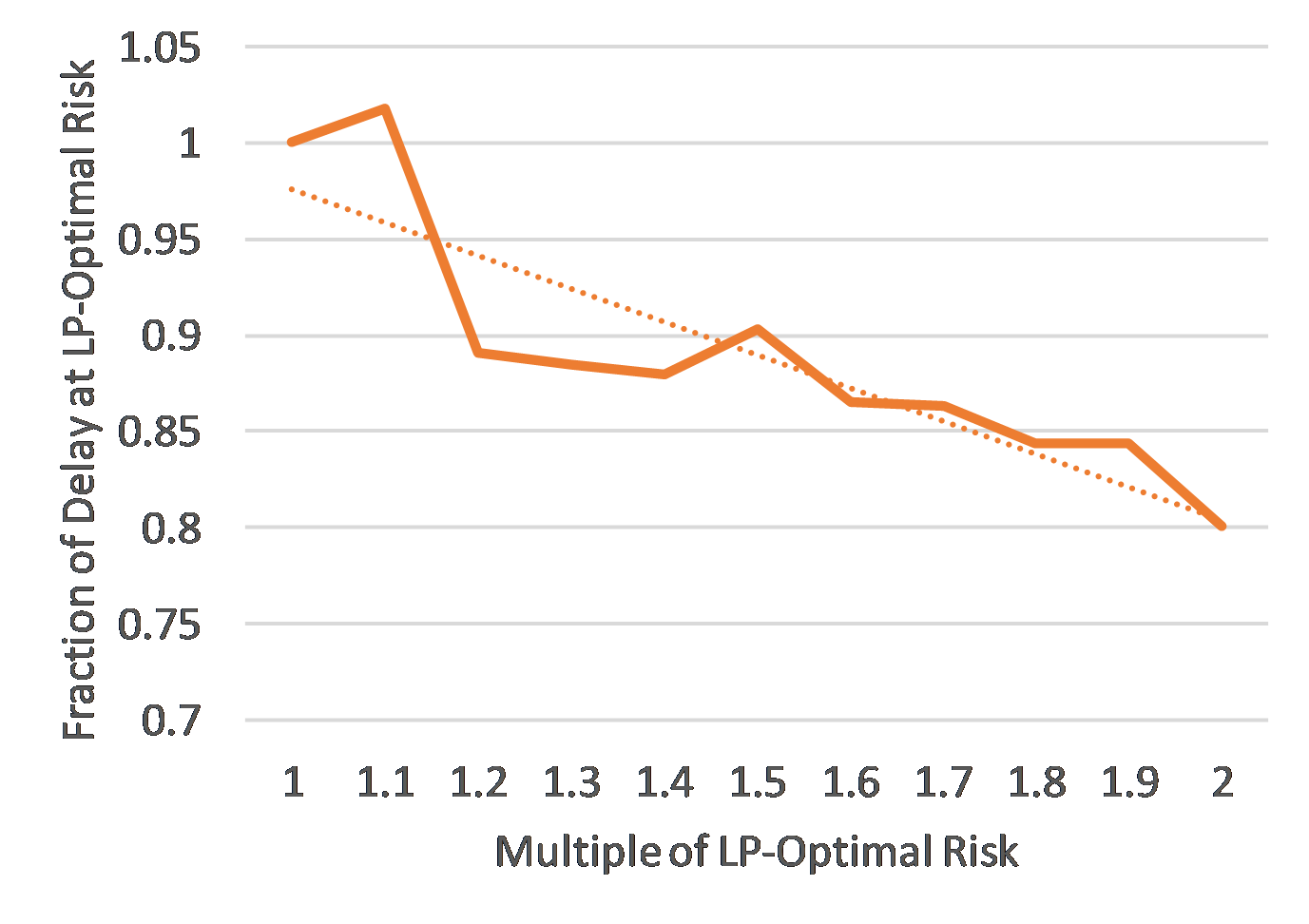}
    \caption{\textbf{Delay vs. Risk}: We measure how the delay decreases (as fraction of delay at LP optimal risk) with increasing risk allowance.}
    \label{fig:delayvsrisk}
\end{figure}
In Figure \ref{fig:delayvsrisk}, we show the inherent trade-off between risk and average delay by varying $\psi$ (keeping the ratio of $\psi_\theta$ fixed across $\theta$) and measuring its effect on the delay. To do this, we take the optimal risk level obtained by solving the LP and then measure the impact on average delay as we relax it. The results show average delay as a fraction of the delay obtained at the LP optimal risk. The resulting curve can be seen as the Pareto frontier described in Section \ref{sect:gametheory} restricted to the case where the ratio of $\psi_\theta$ fixed across $\theta$.  As expected, less stringent risk requirements result in lower average delay.

%Arunesh: I am removing this as this can open another can of worms about how should we choose risk levels.
%This results also shows another important point. In previous approaches, the worst-case risk allowed is determined by the most busy time-windows and the risk for other time windows just have to  but then the other non-busy periods have an associated risk that has more slack than the optimal for that time window. This is another reason that we may do better in the case of multiple time windows.

\subsection{Scalability}
\begin{figure}
    \centering
    \includegraphics[width=\hsize]{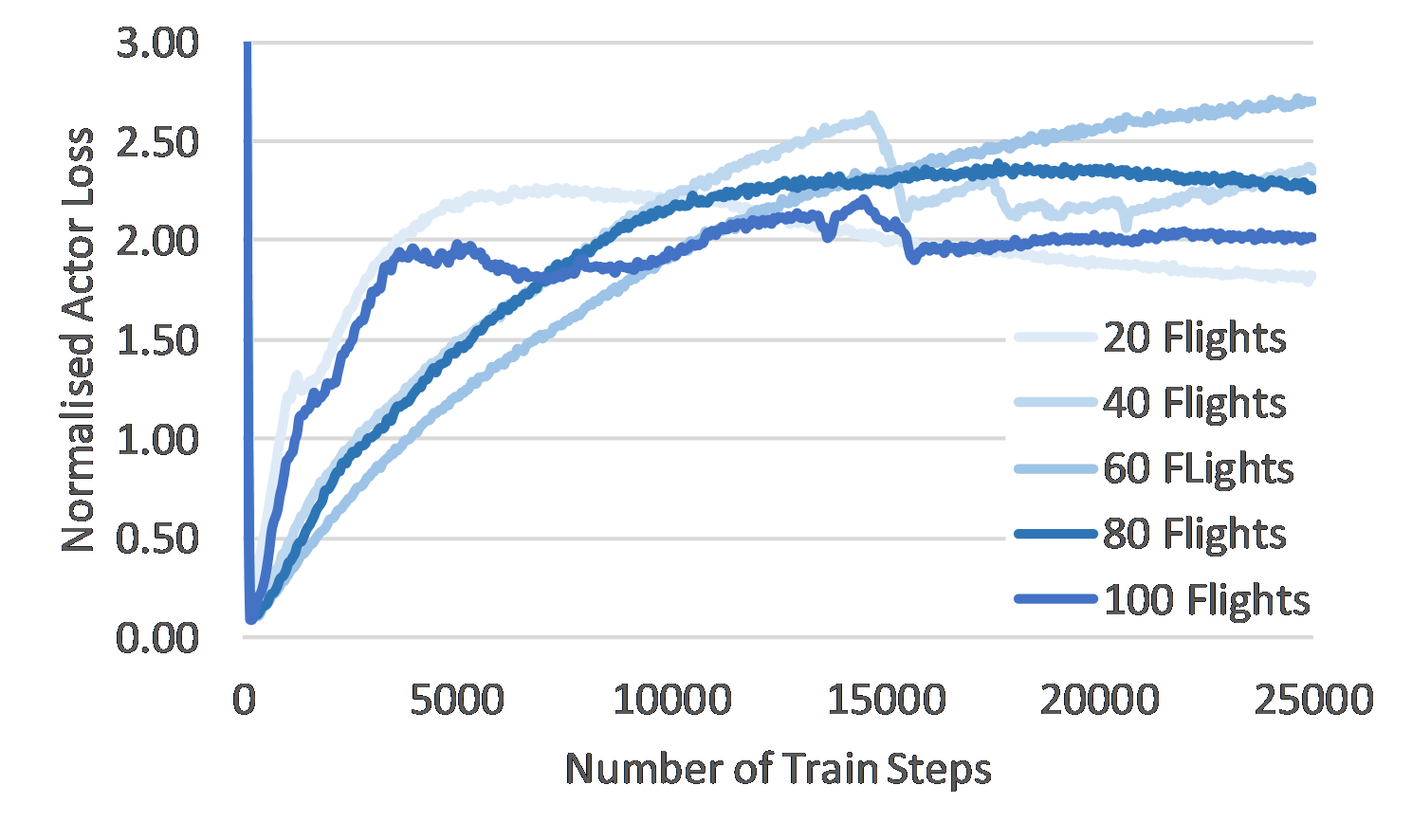}
    \caption{\textbf{Scalability in training steps:} We measure the actor loss of our DDPG network as a function of the number of mini-batches that are fed to it as input.}
    \label{fig:scalevstime}
\end{figure}

\begin{figure}
    \centering
    \includegraphics[width=\hsize]{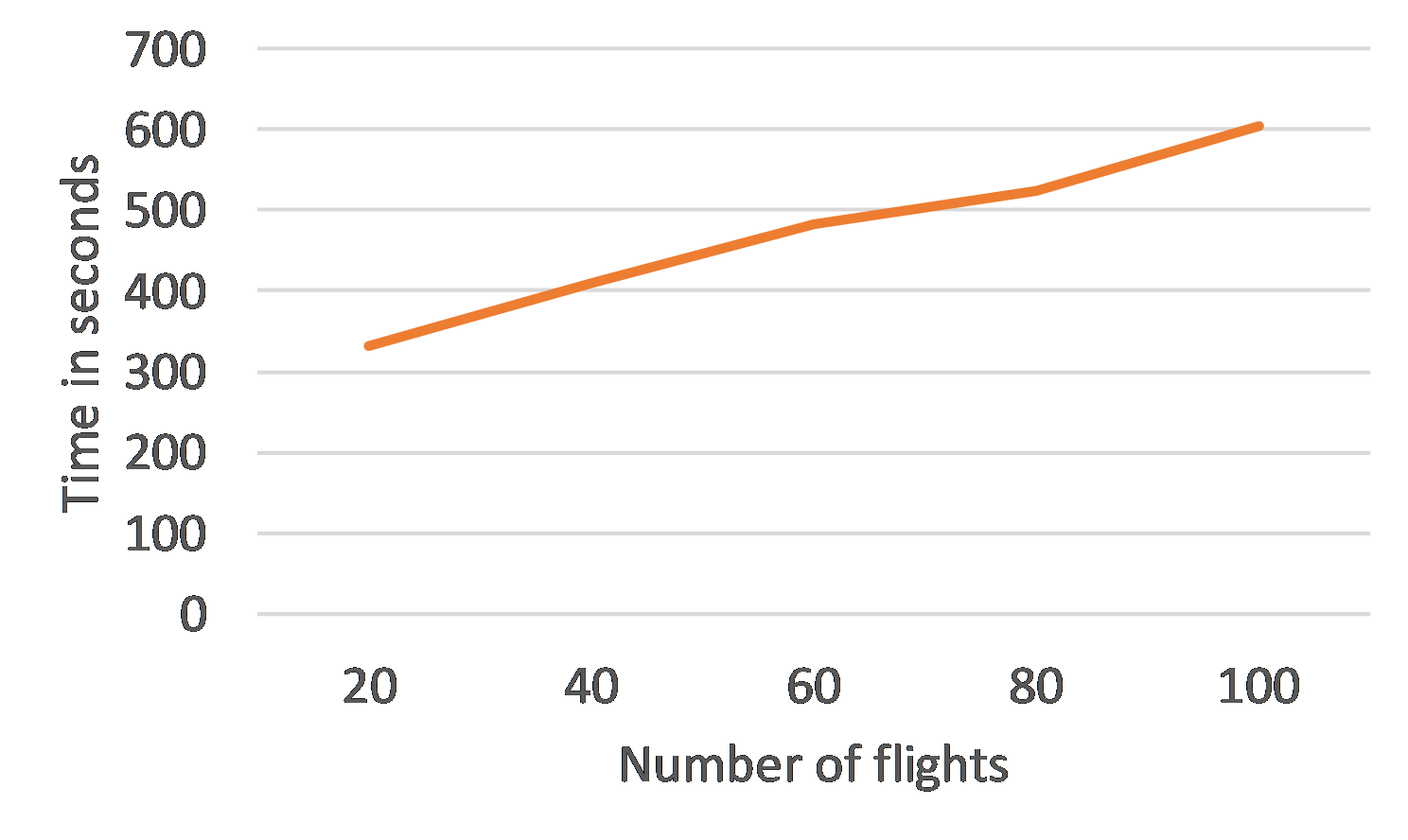}
    \caption{\textbf{Scalability in actual time:} We measure the time in secs at 10,000 training steps for different number of flights.}
    \label{fig:actualtime}
\end{figure}
In Figure \ref{fig:scalevstime}, we show how training time is affected by the size of the game instance. Given that neural networks are not guaranteed optimal, measuring this is slightly challenging but we use the metric of the time at which our DDPG's actor network converges to measure how long the training time takes. As we see in the figure, the number of steps to convergence is about the same, regardless of the number of flights. Based on Figure~\ref{fig:scalevstime}, we choose 10,000 training steps as the number of steps for convergence.

Of course, as the input size increases with increasing number of flights, the time taken per training step increases, thus, the actual wall clock time to get to 10,000 steps for different number of flights varies. Figure~\ref{fig:actualtime} show the actual wall clock time for convergence (to 10,000 steps) with varying number of flights. The increase appears linear showing the scalability of our approach (as a reminder these results are not even using GPUs). In contrast, past work~\cite{brown2016one,mccarthy2017staying} scale highly non-linearly with number of flights and have shown solutions only up to 50 flight and 15 flights respectively. 
%Given a constant hidden layer size, this grows linearly. Secondly, the number of local optima that are reached by our algorithm seems to increases with an increase in number of flights. \textbf{Arunesh: can we additionally show wall clock times here. Sanket: I can simulate it, but I haven't saved it exactly.} 

\subsection{Delay vs. Variance} \label{sec:var}
\begin{figure}
    \centering
    \includegraphics[width=\hsize]{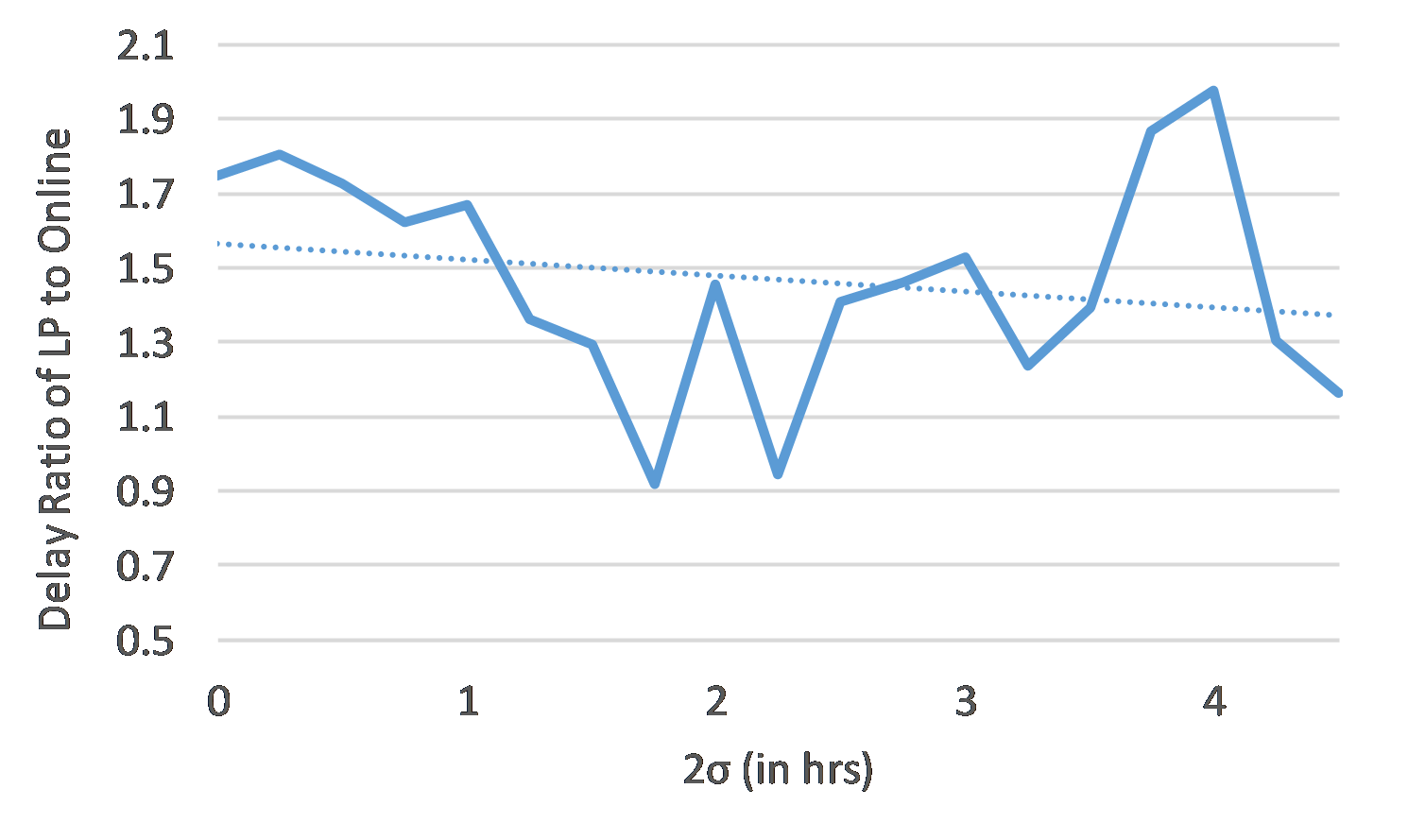}
    \caption{\textbf{Delay vs. Variance:} We measure the ratio of delays for policies computed by LP to our approach across different standard deviation for our arrival distribution.}
    \label{fig:delayvsvariance}
\end{figure}
In Figure \ref{fig:delayvsvariance}, we look at how the variance associated with the passenger arrival distribution affects our gain over the time window based solutions. Here the x-axis is measured using 2$\sigma$ because of the intuition that 95\% of passengers of a flight arrive within a 2 standard deviation window around the mean. This graph can be interpreted as the effect that changing the width of the arrival window (of 95\% passengers) has on solution quality. We vary it from 0 to 5 hours.

We find that as the variance associated with arrivals increases, the gain obtained by using an online approach as ours decreases. We believe that this is because the amount of structure present in the problem decreases as the variance increases. In the limit, when the variance is infinity, the arrivals are uniformly distributed, memory-less, and resemble a Poisson process. As a result, there is no information to be gained whenever the next passenger arrives, hence the per-passenger adaptive solution would do as good as one per-time window adaptive solution in this limiting case. Conversely, if the same number of flights arrive over a longer duration, say 24 hours, our algorithm would do considerably better in terms of the average delay since the arrival windows are less likely to overlap, resulting in smaller overall variance.

\section{Conclusion}
In summary, we proposed a novel model for threat screening that captures inherent features of the problem such as continuous arrival of screenees. We then provided an RL-based method to solve the model which includes the novel $\alpha$-projection method for imposing hard constraints on actions. We believe these advances make our approach for threat screening realistic and applicable in practice.

\subsection{Acknowledgement}
This research was supported by the Singapore Ministry of
Education Academic Research Fund (AcRF) Tier 2 grant MOE2016-T2-1-174 and Ministry of
Education Academic Research Fund (AcRF) Tier 1 grant 19-C220-SMU-011.

\bibliographystyle{aaai}
\bibliography{references}

\end{document}